%% file: qapl2015_LeeDeVink.tex
\documentclass[copyright,creativecommons]{eptcs}

\usepackage{tikz}
\usetikzlibrary{calc,arrows}

\usepackage[utf8]{inputenc}
\usepackage{amsthm}
\usepackage{amssymb}
\usepackage{amsmath}
\usepackage{multirow}
\usepackage{rotating}
\usepackage{textcomp}
\usepackage{color}
\usepackage{fixltx2e}
\usepackage{ifdraft}
\usepackage{stmaryrd}
\usepackage{paralist}
\usepackage{bm}

\usepackage{graphicx}
\usepackage{caption}
\usepackage{subcaption}

\DeclareMathAlphabet{\mathcal}{OMS}{cmsy}{m}{n}

\input{macros} 
\input{macrosPTSS}

\title{Rooted branching bisimulation as a congruence for probabilistic
  transition systems}
\author{%
  Matias D. Lee\thanks{Supported by 2010--2401/001--001--EMA2 and EU\:7FP
    grant agreement 295261 (MEALS) and SECyT--UNC.} 
  \institute{FaMAF, UNC--CONICET \\ 
    C\'ordoba, Argentina}
  \email{lee@famaf.unc.edu.ar}
  \and
  Erik P. de Vink
  \institute{TU/e, Eindhoven, The Netherlands \\
  CWI, Amsterdam, The Netherlands}
\email{evink@win.tue.nl}
}

\providecommand{\event}{QAPL 2015}
\providecommand{\firstpage}{1}
\def\titlerunning{Rooted branching bisimulation as a congruence for PTS}
\def\authorrunning{Lee \& De Vink}
\begin{document}
\maketitle

\begin{abstract}
  We propose a probabilistic transition system specification format,
  referred to as \emph{probabilistic RBB safe}, for which rooted
  branching bisimulation is a congruence.
  The congruence theorem is based on the approach of Fokkink for the
  qualitative case.
  For this to work, the theory of transition system specifications in
  the setting of labeled transition systems needs to be extended to
  deal with probability distributions, both syntactically and
  semantically.
  We provide a scheduler-free characterization of probabilistic
  branching bisimulation as adapted from work of Andova et al.\ for
  the alternating model. 
  Counter examples are given to justify the various conditions
  required by the format. 
\end{abstract}

\setlength{\abovedisplayskip}{2pt}
\setlength{\belowdisplayskip}{2pt}

\input{introduction}

\input{ptss}
\input{bisimulations}

\input{format}

\input{conclusion}

\bibliographystyle{eptcs}
\bibliography{pbb}

\end{document}

%% file: macros.tex
\newtheorem{theorem}{Theorem}
\newtheorem{lemma}[theorem]{Lemma}

\newtheorem{definition}[theorem]{Definition}
\newtheorem{example}[theorem]{Example}

\newcommand{\blankline}{\vspace*{0.25\baselineskip}}

\newcommand{\Sigmasigma}{\Sigma_{\mkern2mu \sigma}}
\newcommand{\lc}{\mathopen{\lbrace \;}}

\newcommand{\pisigmas}{\pi_{\mkern1mu \varsigma \mkern-2mu , s}}

\newcommand{\rc}{\mathclose{\; \rbrace}}
\newcommand{\sigmai}{\sigma_{\mkern-3mu i}}
\newcommand{\sigmaj}{\sigma_{\mkern-4mu j}}
\newcommand{\singleton}[1]{\lbrace {#1} \rbrace}
\newcommand{\thetaj}{\theta_{\mkern-4mu j}}
\newcommand{\transs}[1][]{\overset{\, #1\, }{\Longrightarrow}}

\newcommand{\scdl}{\varsigma}

\newcommand{\A}{\textsf{A}}
\newcommand{\B}{\mathop{\mathcal{B}}}

\newcommand{\calG}{\mathcal{G}}
\newcommand{\calR}{\mathcal{R}}

\makeatletter
\newif\if@@dist
\@@distfalse
\newcommand{\dist}[1]{\@@disttrue\bm{#1}\@@distfalse}
\makeatother

\newcommand{\rest}{\downharpoonright}
\renewcommand{\rest}{, \mkern-1mu}

\newcommand{\frags}{\textit{frags}}
\newcommand{\trace}{\textit{trace}}
\newcommand{\last}{\textit{lst}}
\newcommand{\first}{\textit{fst}}
\newcommand{\support}{\textit{spt}}

\newcommand{\SD}{\textit{SubDisc}}
\renewcommand{\SD}{\textit{SubD}}
\renewcommand{\SD}{\mathcal{SD}}
\renewcommand{\SD}{\Delta_{\textit{sub}}}

\newcommand{\length}{\textit{length}}

\newcommand{\bbisim}{\approx_b}
\newcommand{\nbbisim}{{\not \approx}_b}
\newcommand{\pbbisim}{\approx_{p\!b}}
\newcommand{\npbbisim}{{\not\approx}_{p\!b}}

\newcommand{\rbbisim}{\approx_{\textit{rb}}}
\newcommand{\notrbbisim}{{\not \approx}_{\textit{rb}}}
\newcommand{\nsbisim}{\sim_{b}}

\newcommand{\rel}{\mathrel{\textsf{R}}}
\newcommand{\relR}{\mathrel{\mathcal{C}}}
\newcommand{\relB}{\mathrel{\mathcal{B}}}

\newcommand{\runningPTSS}{\tilde{P}}

\newcommand{\newtext}[1]{#1}

%% file: macrosPTSS.tex
\newcommand{\dedrule}[2]{\frac{#1}{#2}}
\newcommand{\trans}[1][]{\xrightarrow{\, {#1} \, }}

\newcommand{\ntrans}[1][]{\mathrel{{\trans[#1]}\makebox[0em][r]{$\not$\hspace{2ex}}}{\!}}

\newcommand{\gtgeq}{\trianglerighteq}

\newcommand{\strans}[1][]{\stackrel{#1}{\longrightarrow}}
\newcommand{\tuple}[1]{\langle{#1}\rangle}

\newcommand{\rank}{\mathop{\textsl{rk}}}
\newcommand{\arity}{\mathop{\textsl{ar} \mkern1mu}}

\newcommand{\openT}{\mathbb{T}}

\newcommand{\openTerms}[1][]{\openT_{ #1 }(\Sigma)}
\newcommand{\openSTerms}{\openT(\Ssignature)}
\newcommand{\openDTerms}{\openT(\Dsignature)}
\newcommand{\closedSTerms}[1][]{T(\Sigma_{ \TS })}
\newcommand{\closedDTerms}[1][]{T(\Sigma_{ \DTS })}
\newcommand{\closedTerms}{T(\Sigma)}

\newcommand{\Var}{\mathop{\textsl{var}}}

\newcommand{\prem}[1]{\textsl{prem} \mkern1mu (#1)}
\newcommand{\conc}[1]{\textsl{conc} \mkern1mu (#1)}

\newcommand{\Red}[2]{\textrm{Red}}

\newcommand{\TVar}{\mathcal{V}}
\newcommand{\DVar}{\mathcal{D}}

\newcommand{\PTSS}{\mbox{\emph{PTSS}}}
\newcommand{\Tr}{\textsf{Tr}}
\newcommand{\CerTr}{{\textsf{CT}}}
\newcommand{\PosTr}{{\textsf{PT}}}

\newcommand{\nullproc}{0}

\newcommand{\Act}{A}
\renewcommand{\Act}{\mathcal{A}}

\newcommand{\ntmufxt}{\ensuremath{
  \textsf{n} \mkern0.5mu \textsf{t} \mkern0.5mu \mu \mkern0.5mu 
  \textsf{f} \mkern0.5mu \theta
  \mkern-1.5mu / \mkern-1.75mu 
  \textsf{n} \mkern0.5mu \textsf{t} \mkern0.5mu \mu \mkern0.5mu 
  \textsf{x} \mkern0.5mu \theta}}

\newcommand{\sem}[2]{\llbracket{#1}\rrbracket_{#2}}

\newcommand{\proves}{\vdash}

\newcommand{\TS}{s}
\newcommand{\DTS}{d}

\newcommand\restr[2]{{
  \left.\kern-\nulldelimiterspace 
  #1 
  \right|_{#2} 
  }}

\newcommand{\myparagraph}[1]{\addvspace{\medskipamount}\noindent\textbf{#1} }

\newcommand{\trgt}[1]{\text{tgt}(#1)}

\newtheorem*{theorem*}{Theorem}

\newcommand{\Ssignature}{\Sigma_\TS}
\newcommand{\Dsignature}{\Sigma_\DTS}

%% file: introduction.tex
\section{Introduction}

Structural operational semantics is a standard methodology to provide a
semantics to programming languages and process
algebras~\cite{AFV99:handbook}.  
%
In this setting, a signature~$\Sigma$, a set of actions~$\Act$ and a
set of rules~$\calR$ define a transition system specification~(TSS).
The signature is used to define the terms of the language~$T(\Sigma)$.
The interpretation of each term is given by a labeled transition system
where the states range over~$T(\Sigma)$, the labels over~$\Act$ 
and the transitions are governed by the set of rules~$\calR$.
This technique has been widely studied in the context of qualitative
process algebras, see~\cite{MRG07} for an overview.
A particular topic of interest is the study of rule formats for which
a behavioral equivalence is guaranteed to be a congruence.
This property, crucial for a compositional analysis, ensures that
pairwise equivalent terms $t_1, \ldots, t_n$ and $t'_1, \ldots, t'_n$
give rise to the equivalence of the terms $f(t_1,\ldots t_n)$ and
$f(t'_1, \ldots, t'_n)$,

Much work has been done on formats for labeled transition systems,
i.e.\ in a qualitative setting.
Although the main focus is on strong bisimulation
\cite{GV92,Fok94:tacs,AFV99:handbook}, weaker notions that support
abstraction from internal behavior has been studied as well.
In \cite{Bloom95}, Bloom introduced formats covering weak
bisimulation, rooted weak bisimulation, branching bisimulation and
rooted branching bisimulation.
In \cite{vanG05} these results are extended for $\eta$-bisimulation
and delay bisimulations.  
In addition, in~\cite{Fok00}, a more liberal format is introduced for
rooted branching bisimulation, called RBB~safe.
%
%
%
%

Going into the other direction, addition of quantitative information,
e.g.\ the inclusion of probabilities, allows to model systems in more
detail. 
%
%
The aim to include these features motivated further extension of the
theory of structured operational semantics.
In the context of probabilistic process algebras previous results in
the qualitative setting have been extended to deal with quantitative
decorations. Based on the category-theoretical framework of Turi
and~Plotkin \cite{Bart04}~proposes \emph{probabilistic GSOS}. In turn,
this has been generalized in~\cite{Kli09:mosses,KS13:ic} yielding
\emph{weighted GSOS}. For generative systems \cite{LT09}~introduces
\emph{GTTS}, a format allowing features like look-ahead. The format
\emph{ntyft/ntyxt} is translated to the probabilistic setting
in~\cite{DL12}, while generation of
axiomatizations~\cite{ABV94:ic,BV04:jlap} for probabilistic GSOS is
reported in~\cite{DGL14}. In~\cite{MP14} so-called
\emph{weight-function SOS} is proposed for \emph{ULTraS}, a
generalization of probabilistic transition systems.
The bottom line for all these works
is that strong bisimulation is a congruence for probabilistic
transition systems.

So far, no congruence format is given dealing with both additional
aspects, i.e.\ taking quantitative information into account, and, at
the same time, catering for a weak notion of bisimulation.
To the best of our knowledge, this paper presents the first
specification format for the quantitative setting that respects a weak
equivalence, viz.\ branching bisimulation.
Our work follows the approach initiated in~\cite{DL12,LGD12}
to define probabilistic transitions systems specifications ($\PTSS$). 
A~$\PTSS$ is composed of a two-sorted signature $(\Ssignature,
\Dsignature)$, a set of actions $\Act$ that contains the internal
action~$\tau$, and a set of rules~$\calR$.
The signature~$\Ssignature$ is used to represent states, while
$\Dsignature$~is used to represent distributions over states.
Transitions have the shape $t \trans[a] \theta$ with the following
intended meaning: state~$t$ can execute an action~$a$ and the next
state is selected using the state distribution~$\sem{\theta}{}$.
Thus, the interpretation of the distribution term~$\theta$ is a
probability distribution over states.
Note, this framework allows for non-determinism: if $t \trans[a]
\theta$ and $t \trans[a] \theta'$ we do not require~$\theta =
\theta'$.
Hence, a $\PTSS$ defines systems similar to Segala's probabilistic
automata~\cite{Seg95a}.
In order to prove that branching bisimulation is a congruence, we
follow the ideas presented in~\cite{Fok00}.
For that reason, we call our specification format the
\emph{probabilistic RBB safe} format, referring to the RBB safe format
proposed by Fokkink.
Often, working with probabilistic automata where non-determinism and
probabilities coexist, requires the introduction of schedulers to
determine the probability of a particular
execution~\cite{SL95,EHKTZ2013}.
A key point for our proof of the main result is a characterization of
branching bisimulation without schedulers.
This characterization was introduced for the alternating model
in~\cite{AGT2012}, and we have adapted it to our context.

In Section~\ref{sec:preliminaries}, two-sorted signatures and
distributions over states are introduced. 
Section~\ref{sec:ptss} introduces probabilistic transition system
specifications. 
Section~\ref{section:bisimulations} discusses schedulers, weak
transitions, probabilistic branching bisimulation and branching bisimulation. 
\newtext{
The second relation is a restricted version of the former. 
The format that we present deals with branching bisimulation.}
In Section~\ref{sec:format}, the probabilistic RBB safe specification
format is presented. We furthermore touch upon some intricacies of it
and give the congruence theorem.
In Section~\ref{sec:the_end} some conclusions are drawn and
future lines of research are indicated.

%% file: ptss.tex


\section{Preliminaries}
\label{sec:preliminaries}

We fix the set $S = \singleton{ \TS, \mkern1mu \DTS \mkern1mu }$
denoting two sorts. Elements of sort~$\TS$ are intended to represent
states in a transition system, while elements of sort~$\DTS$ will
represent distributions over states. We let $\sigma$ range
over~$S$. 

An $S$-sorted signature is a structure $( \Sigma , \arity )$, where
\begin{inparaenum}[(i)]
\item $\Sigma$~is a set of function names, and
\item $\arity \colon \Sigma \to S^\ast \times S$ is the arity
  function.
\end{inparaenum}
The rank of $f \in \Sigma$ is the number of arguments of~$f$,
i.e.\ $\rank(f) = n$ if $\arity(f) = \sigma_1 \cdots \sigma_n \to
\sigma$. We write $\sigma_1 \cdots \sigma_n \to \sigma$ instead of
$(\sigma_1 \cdots \sigma_n,\sigma)$ to highlight that function~$f$
maps to sort~$\sigma$. Function~$f$ is a constant if $\rank(f) =
0$. To simplify the presentation we will write an $S$-sorted
signature~$\Sigma$ as a pair of disjoint signatures $( \Ssignature,
\Dsignature )$ where $\Ssignature$ is the set of operations that map
to sort~$\TS$, and $\Dsignature$ is the set of operations that map to
sort~$\DTS$.

Let $\TVar$ and~$\DVar$ be two infinite sets of $S$-sorted variables,
for states and distributions, respectively,
where $\TVar$, $\DVar$, and~$\Sigma$ are all mutually disjoint. We use
$x, y, z$ (with possible decorations, subscripts or superscripts) to
range over~$\TVar$, $\mu,\nu$~to range over~$\DVar$ and $\zeta$~to
range over~$\TVar \cup \DVar$.

\medskip

\begin{definition} 
  \label{def:state_distribution_terms} 
  Let $V \subseteq \TVar$ and $D \subseteq \DVar$. We define the sets
  of \emph{state terms} $T(\Ssignature,V,D)$ and \emph{distribution
    terms} $T(\Dsignature,V,D)$ as the smallest sets of terms
  satisfying
  \begin{enumerate}[(i)]
  \item $V \subseteq T(\Ssignature, V, \mkern1mu D)$, and $D \subseteq
    T(\Dsignature, V, \mkern1mu D)$;
  \item $f(\xi_1, \cdots, \xi_{\rank(f)}) \in T(\Sigmasigma, V,
    \mkern1mu D)$ if
    $\arity(f) = \sigma_1 \cdots \sigma_n \to \sigma$ and $\xi_i \in
    T(\Sigma_{\mkern2mu \sigma \mkern-2mu {}_i},V, \mkern1mu D)$.
  \end{enumerate}
\end{definition}

\blankline

\noindent
We let $\openTerms = T(\Ssignature,\TVar,\DVar) \cup
T(\Dsignature,\TVar,\DVar)$ denote the set of all \emph{open terms}
and distinguish the sets $\openSTerms = T(\Ssignature,\TVar,\DVar)$ of
\emph{open state terms} and $\openDTerms = T(\Dsignature,\TVar,\DVar)$
of \emph{open distribution terms}.
Similarly, we let $\closedTerms = T(\Ssignature,\emptyset,\emptyset)
\cup T(\Dsignature,\emptyset,\emptyset)$ denote the set of all
\emph{closed terms} and distinguish the sets $\closedSTerms =
T(\Ssignature,\emptyset,\emptyset)$ of \emph{closed state terms} and
$\closedDTerms = T(\Dsignature,\emptyset,\emptyset)$ of \emph{closed
  distribution terms}.
We let $t$, $t'$, $t_1$,\ldots\ range over state terms, $\theta$,
$\theta'$, $\theta_1$,\ldots\ range over distribution terms, and
$\xi$, $\xi'$, $\xi_1$,\ldots\ range over any kind of terms.
We use $\Var(\xi) \subseteq \TVar \cup \DVar$ to denote the set of
variables occurring in term~$\xi$.

Let $\Delta(\closedSTerms)$ denote the set of all (discrete)
probability distributions over~$\closedSTerms$.
We let $\pi$ range over $\Delta(\closedSTerms)$.
For each $t \in \closedSTerms$, let $\delta_t \in
\Delta(\closedSTerms)$ denote the \emph{Dirac distribution}, i.e.,
$\delta_t(t)=1$ and $\delta_t(t')=0$ if $t$ and $t'$ are not
syntactically equal.
For $X \subseteq \closedSTerms$ we define $\pi(X) = \sum_{t\in X} \:
\pi(t)$. The convex combination $\sum_{i \in I} \: p_i \pi_i$ of a
(finite) family $\{\pi_i\}_{i \in I}$ of probability distributions
with $p_i \in (0,1]$ and $\sum_{i \in I} p_i = 1$ is defined by
$(\sum_{i \in I} \: p_i \pi_i)(t) = \sum_{i \in I} \: (p_i \pi_i(t))$.

The type of signatures we consider are of a particular form, to which
we refer to as ``probabilistic lifted signature''.  We start from a
signature~$\Ssignature$ of functions mapping into sort $\TS$ and
construct the signature~$\Dsignature$ of functions mapping into~$\DTS$
as follows.
For each~$f \in \Ssignature$ we include a function symbol $\dist{f}
\in \Dsignature$ with $\arity(\dist{f}) = \DTS \cdots \DTS \to \DTS$
and $\rank(\dist{f}) = \rank(f)$. We call $\dist{f}$ the
\emph{probabilistic lifting} of~$f$. (We use boldface fonts to
indicate that a function in $\Dsignature$ is the probabilistic lifting
of another in $\Ssignature$.)
Moreover $\Dsignature$ may include any of the following additional
operators:
\begin{itemize}
\item $\delta$ with arity $\arity(\delta) = \TS\to\DTS$ and

\item $\bigoplus_{i\in I}[p_i]\textvisiblespace$
  of arity
  $\arity\left(\bigoplus_{i\in I}[p_i]\textvisiblespace\right) =
  \DTS^{|I|} \to \DTS$
  for index set~$I$, 
  $p_i\in(0,1]$ for~$i\in I$,
  $\sum_{i \in I} p_i = 1$.
\end{itemize}
In particular, we write $\theta_1\oplus_{p_1}\theta_2$ instead of
$\bigoplus_{i\in\{1,2\}}[p_i]\theta_i$ when the index set~$I$ has
exactly two elements.

The operators $\delta$ and $\bigoplus_{i\in I}[p_i]\textvisiblespace$
are used to construct discrete probability functions of finite
support:
$\delta(t)$ is interpreted as the Dirac distribution for~$t$, and
$\bigoplus_{i\in I}[p_i]\theta_i$ represents a distribution that weights
with $p_i$ the distribution represented by the term $\theta_i$.
Moreover, a probabilistically lifted operator $\dist{f}$ is
interpreted by properly lifting the probabilities of the operands to
terms composed with the operator $f$.

Formally, the algebra associated with a probabilistically lifted
signature $\Sigma = (\Ssignature,\Dsignature)$ is defined as follows.
For sort $\TS$, it is the freely generated algebraic structure
$\closedSTerms$.  
For sort $\DTS$, it is defined by the carrier
$\Delta(\closedSTerms)$ and the following interpretation:
\begin{itemize}
  \item%
    $\sem{\delta(t)}{} = \delta_t$ for all $t\in \closedDTerms$
  \item%
    $\sem{\bigoplus_{i\in I}[p_i]\theta_i}{} =
     \sum_{i\in I} \: p_i \sem{\theta_i}{}$
    for $\lc \theta_i \mid i\in I \rc \subseteq \closedDTerms$

  \item%
    $\displaystyle
     \sem{\dist{f}( \theta_1, \ldots, \theta_{\rank(f)})}{}(g(\xi_1,
     \ldots, \xi_{\rank(g)})) =
       \begin{cases}
         \prod_{\sigmai=s} \: \sem{\theta_i}{}(\xi_i) 
            & \text{if $f \equiv g$ and for all $j$: $\sigmaj = d
              \Rightarrow \xi_j = \thetaj$} \smallskip \\
         0  & \text{otherwise}
       \end{cases}$
\end{itemize}
In the above definition we assume that $\prod\emptyset = 1$.
Notice that in the semantics of a lifted function $\dist{f}$, the
product at the right-hand side only considers the distributions
related to the \TS-sorted positions in $f$, while the distribution
terms corresponding to the \DTS-sorted positions in $f$ should match
exactly the parameters of~$\dist{f}$.

A substitution~$\rho$ is a map $\TVar \cup \DVar \to \openTerms$ such
that $\rho(x) \in \openSTerms$, for all $x \in \TVar$, and $\rho(\mu)
\in \openDTerms$, for all $\mu \in \DVar$.  A substitution is closed
if it maps each variable to a closed term. A substitution extends to a
mapping from terms to terms as usual.

\blankline

\begin{example}\label{ex:running:sig}
  To set the context of our running example we introduce the following
  signature.
  We assume a set $\Act$ of action labels, and distinguish $\tau \in \Act$.
  Let $\Sigma = ( \Ssignature, \Dsignature )$ be a probabilistically
  lifted signature where~$\Ssignature$ contains:
 %
   \begin{inparaenum}[(i)] 
   \item%
     constants $\nullproc$ (inaction or the stop process) of sort
     $\TS$, i.e., $\arity(\nullproc) = \TS$;
   \item%
     a family of unary probabilistic prefix operators
     $a.\textvisiblespace$ for~$a \in \Act$ with $\arity(a) = \DTS \to
     \TS$;
   \item%
     and a binary operator~$+$ (alternative composition or sum) 
     with  $\arity({+}) = {\TS}{\TS} \to \TS$.
 \end{inparaenum}
   Moreover, $\Dsignature$ contains $\delta$, all binary operators
   $\oplus_p$, and the lifted operators, which are as follows:
   \begin{inparaenum}[(i)] 
   \setcounter{enumi}{3}
   \item%
     the constant $\dist{\nullproc}$ with $\arity(\dist{\nullproc}) = \DTS$;
   \item%
    the family of unary operators $\dist{a.}\textvisiblespace$ for~$a
    \in \Act$ with $\arity(\dist{a}) = \DTS \to \DTS$;
   \item%
    the binary operator~$\dist{+}$ with $\arity(\dist{+}) =
    {\DTS}{\DTS} \to \DTS$.
   \end{inparaenum}

   The intended meaning of the probabilistic prefix
   operator~$a.\theta$ is that this term can perform action~$a$ and
   move to a term~$t$ with probability~$\sem{\theta}{}(t)$.
   The stop process~$0$ and alternative composition ${+}$ have their
   usual meaning.
\end{example}

\section{Probabilistic Transition System Specifications}
\label{sec:ptss} 
 
In our setting, a probabilistic transition relation prescribes what
possible activity can be performed by a term in a signature. Such
activity is described by an action and a probability distribution on
terms that indicates the probability to reach a particular new term.
Our definition is along the lines of probabilistic
automata~\cite{Seg95a}.

\begin{definition}[PTS]%
  A \emph{probabilistic labeled transition system} (PTS) is a triple
  $\A = ( \mkern1mu \closedSTerms, \, \Act, \, {\trans} \mkern1mu )$,
  where $\Sigma = (\Ssignature,\Dsignature)$ is a probabilistically
  lifted signature, $\Act$ is a set of actions, and
  ${\trans} \subseteq \closedSTerms \times \Act \times
  \Delta(\closedSTerms)$ is a transition relation.
\end{definition}

  

\blankline

\noindent
We build on~\cite{GV92,Gro93,BG96} for our definition of a
probabilistic transition system specification.
 
\begin{definition}[PTSS]\label{def:ptss}%
  A \emph{probabilistic transition system specification} (PTSS) is a
  triple $P = ( \mkern1mu \Sigma, \, \Act, \, R \mkern1mu )$ where
  $\Sigma$~is a probabilistically lifted signature, $\Act$~is a set of
  labels, and $\calR$ is a set of rules of the form:
  \begin{displaymath}
    \dedrule{ \lc t_k \trans[a_k] \theta_k \mid k \in K \rc 
      \; \cup \;
      \lc t_\ell \ntrans[b_\ell] \mid \ell \in L \rc}{t \trans[a]
      \theta} 
  \end{displaymath}
   where $K, L$ are index sets, $t, t_k, t_\ell \in \openSTerms$, $a,
   a_k, b_\ell \in \Act$, and $\theta_k, \theta \in \openDTerms$.
 \end{definition}
 
\blankline

\noindent
Expressions of the form $t \trans[a] \theta$ and $t \ntrans[a]$ are
called \emph{positive literals} and \emph{negative literals},
respectively.
For any rule $r \in \calR$, literals above the line are called
\emph{premises}, notation $\prem{r}$; the literal below the line is
called the \emph{conclusion}, notation $\conc{r}$.
%
%
%
%
Substitutions provide instances to the rules of a PTSS that, together
with some appropriate machinery, allow us to define probabilistic
transition relations. Given a substitution~$\rho$, it extends to
literals as follows:
%
  $\rho( \mkern1mu t \trans[a] \theta \mkern1mu ) = \rho(t) \trans[a]
  \rho(\theta)$ 
  and  
  $\rho( \mkern1mu t \ntrans[a]) = \rho(t) \ntrans[a]$.
We say that $r'$~is a (closed) instance of a rule~$r$ if there is a
(closed) substitution~$\rho$ such that~$r' = \rho(r)$.
  
For the sake of clarity, in the rest of the paper, we will deal with models as \emph{symbolic}
transition relations, i.e.\ subsets of $\closedSTerms \times \Act
\times \closedDTerms$ rather than \emph{concrete} transition relations
in $\closedSTerms \times \Act \times \Delta(\closedSTerms)$ required
by a PTS\@.
We will mostly refer with the term ``transition relation'' to a
symbolic transition relation.
In any case, a symbolic transition relation induces a unique concrete
transition relation by interpreting every target distribution term as
the distribution it defines. That is, the symbolic transition $t
\trans[a] \theta$ is interpreted as the concrete transition $t
\trans[a] \sem{\theta}{}$.  If the symbolic transition relation turns
out to be a model of a PTSS~$P$, we say that the induced concrete
transition relation defines a PTS associated to~$P$.
 
However, first we need to define an appropriate notion of a model. As
has been argued elsewhere (e.g.~\cite{Gro93,BG96,vG04}), transition
system specifications with negative premises do not uniquely define a
transition relation and different reasonable techniques may lead to
incomparable models.
For instance, the PTSS with the single constant~$f$, set of labels
$\singleton{a,b}$ and the two rules
%
  $\frac{ f \ntrans[b]}{ f \trans[a] \dist{f}}$
  and 
  $\frac{ f \ntrans[a]}{ f \trans[b] \dist{f}}$, 
%
has two models that are justifiably compatible with the rules (so
called supported models~\cite{BIM95,BG96,vG04}), viz.\ 
$\{f\trans[a]\dist{f}\}$ and
$\{f\trans[b]\dist{f}\}$.
 
An alternative view is to consider so-called \emph{3-valued models}. A
3-valued model partitions the set $\closedSTerms \times \Act \times
\closedDTerms$ in three sets containing, respectively, the transitions
that are known to hold, that are known not to hold, and those whose
validity is unknown.
Thus, a 3-valued model can be presented as a pair
$\tuple{\CerTr,\PosTr}$ of transition relations $\CerTr, \PosTr
\subseteq {\closedSTerms \times \Act \times \closedDTerms}$, with
$\CerTr \subseteq \PosTr$, where $\CerTr$ is the set of transitions
that \emph{certainly} hold, and $\PosTr$ is the set of transitions
that \emph{possibly} hold.  So, transitions in $\PosTr \setminus
\CerTr$ are those whose validity is unknown and transitions in
$(\closedSTerms \times \Act \times \closedDTerms) \setminus \PosTr$
are those that certainly do not to hold.
In view of the above, a 2-valued model satisfies $\CerTr = \PosTr$.
 
A 3-value model $\tuple{\CerTr,\PosTr}$ that is justifiably compatible
with the proof system defined by a PTSS~$P$ is said to be \emph{stable}
for~$P$.
We will make clear what we mean by ``justifiably compatible'' in
Definition~\ref{def:3-valued-stable-model}.
Before formally defining the notions of a proof and 3-valued stable
model we introduce some notation.
Given a transition relation
$\Tr \subseteq \closedSTerms \times \Act \times \closedDTerms$,
$t \trans[a]\theta$ \emph{holds in} $\Tr$, notation $\Tr \models {t
  \trans[a]\theta}$, if ${t \trans[a] \theta} \in \Tr$;
$t \ntrans[a]$ \emph{holds in} $\Tr$, notation $\Tr \models {t
  \ntrans[a]}$, if for all $\theta\in\closedDTerms$, ${t \trans[a]
  \theta} \notin \Tr$.
Given a set of literals~$H$, we write $\Tr \models H$ if for all $\psi
\in H$, $\Tr \models \psi$.
 
\begin{definition}[Proof]
  \label{def:proof}
  Let $P = (\Sigma, \, \Act, \, \calR)$ be a PTSS\@.
  Let $\psi$ be a positive literal and let $H$ be a set of literals. 
   A \emph{proof} of a transition rule $\frac{H}{\psi}$ from~$P$ is a
   well-founded, upwardly branching tree where each node is a literal
   such that:
   \begin{enumerate}
   \item%
     the root is $\psi$, and 
   \item%
     if $\chi$ is a node and $K$ is the set of nodes directly above
     $\chi$, then one of the following conditions holds:
     \begin{enumerate}
     \item%
       $K = \emptyset$ and $\chi \in H$, or
     \item%
       $\frac{K}{\chi}$ is a valid substitution instance of a rule
       from~$\calR$.
     \end{enumerate}
   \end{enumerate}
   $\dedrule{H}{\psi}$ \emph{is provable} from~$P$, notation $P
   \proves \dedrule{H}{\psi}$, if there exists a proof of
   $\dedrule{H}{\psi}$ from~$P$.
\end{definition}

\blankline

\noindent 
Above we stated that a 3-value stable model $\tuple{\CerTr,\PosTr}$
for a PTSS~$P$ has to be \emph{justifiably compatible} with the proof
system defined by~$P$.  By `compatible' we mean that
$\tuple{\CerTr,\PosTr}$ has to be consistent with every provable rule.
With `justifiable' we require that for each transition in $\CerTr$
and~$\PosTr$ there is actually a proof that justifies it.
More precisely, we require that
\begin{inparaenum}[(a)]
\item%
  for every certain transition in~$\CerTr$ there is a proof in~$P$
  such that all negative hypotheses of the proof are known to hold
  (i.e.\ there is no possible transition in~$\PosTr$ denying a negative
  hypothesis), and
\item%
  for every possible transition in~$\PosTr$ there is a proof in~$P$
  such that all negative hypotheses possibly hold (i.e.\ there is no
  certain transition in~$\CerTr$ denying a negative hypothesis).
 \end{inparaenum}
This is formally stated in the next definition.

\blankline

\begin{definition}[3-valued stable model]
  \label{def:3-valued-stable-model}%
  Let $P = ( \mkern1mu \Sigma, \, \Act, \, \calR \mkern1mu )$ be a
  PTSS\@. 
  A tuple $\tuple{\CerTr,\PosTr}$ with $\CerTr \subseteq \PosTr
  \subseteq{\closedSTerms \times \Act \times \closedDTerms}$ is a
  \emph{3-valued stable model} for~$P$ if for every closed positive
  literal~$\psi$ it holds that
  \begin{enumerate}[(a)]
  \item%
    $\psi \in \CerTr$ iff there is a set $N$ of closed negative
    literals such that $P \proves \dedrule{N}{\psi}$ and
    $\PosTr\models N$
  \item%
     $\psi\in\PosTr$ iff there is a set $N$ of closed negative literals
     such that $P \proves \dedrule{N}{\psi}$ and $\CerTr\models N$.
  \end{enumerate}
\end{definition}
 
\blankline

\noindent
 In fact, the least 3-valued stable model of a PTSS can be constructed
 using induction. We borrow this construction from~\cite{Fok00,FV98}.
 
 \begin{lemma}\label{lem:inductive-construction-of-3v-least-model}
   Let $P$ be a PTSS\@.
   For each ordinal~$\alpha$ define the pair
   $\tuple{\CerTr_\alpha,\PosTr_\alpha}$ as follows:
   \begin{itemize}
   \item%
     $\CerTr_0 = \emptyset$ and
     $\PosTr_0 = \closedSTerms \times \Act \times \closedDTerms$. 
   \item%
     For every non-limit ordinal~$\alpha > 0$, define
     \begin{align*}
       \CerTr_\alpha & \textstyle
       = 
       \lc t \trans[a] \theta \mid
       \text{for some set~$N$ of negative literals, }
       P \proves \dedrule{N}{t \strans[a] \theta}
       \text{ and }
       \PosTr_{\alpha-1} \models N 
       \rc
       \\
       \PosTr_\alpha & \textstyle
       = 
       \lc t \trans[a] \theta \mid
       \text{for some set~$N$ of negative literals, }
       P \proves \dedrule{N}{t \strans[a] \theta}
       \text{ and }
       \CerTr_{\alpha-1}\models N 
       \rc
     \end{align*}
   \item%
     For every limit ordinal~$\alpha$, define
     $\CerTr_\alpha = \bigcup_{\beta<\alpha} \: \CerTr_\beta$ and
     $\PosTr_\alpha = \bigcap_{\beta<\alpha} \: \PosTr_\beta$.
   \end{itemize}
   Then it holds that
   \begin{itemize}
   \item \label{lem:inductive-construction-of-3v-least-model:inclusion}%
     If $\beta \leqslant \alpha$, $\CerTr_\beta\subseteq\CerTr_\alpha$ and
     $\PosTr_\beta \supseteq \PosTr_\alpha$, that is,
     $\tuple{\CerTr_\beta,\PosTr_\beta}$ has at most as much
     information as $\tuple{\CerTr_\alpha,\PosTr_\alpha}$.
   \item\label{lem:inductive-construction-of-3v-least-model:least}%
     There is an ordinal~$\lambda$ such that $\CerTr_\lambda =
     \CerTr_{\lambda+1}$ and $\PosTr_\lambda = \PosTr_{\lambda+1}$.
     Moreover, $\tuple{\CerTr_\lambda,\PosTr_\lambda}$ is the least
     3-valued stable model for~$P$.
   \end{itemize}
 \end{lemma}
 
\blankline

\noindent
This result is shown in~\cite{FV98,Fok00} for a non-probabilistic
setting using a slightly different definition of 3-valued models.
However, with minor changes the same proof applies to our setting as
well.
We note that the first item of the lemma can be proved using
transfinite induction on the lexicographic order of the
pair~$(\alpha,\beta)$.  the second item follows from the
Knaster-Tarski theorem.
PTSS with a least 3-valued stable model that are also a 2-valued model
are of particular interest, since such a model is actually the only
3-valued stable model~\cite{BG96,vG04}.

\blankline

\begin{definition}
  \label{def:complete-PTSS}%
  A PTSS~$P$ is said to be \emph{complete} if its least 3-valued
  stable model $\tuple{\CerTr,\PosTr}$ satisfies that $\CerTr = \PosTr$
  i.e., the model is also 2-valued.
 \end{definition}

 
\blankline

\noindent
Now, we can associate a probabilistic transition system to a complete
PTSS\@.

\blankline

\begin{definition}
  \label{def:associated-model}
  Let $P$ be a complete PTSS and let $\tuple{\Tr,\Tr \mkern2mu }$ be
  its unique 3-valued stable model.
  We say that $\Tr$~is the transition relation associated to~$P$.  We
  also define the PTS associated to~$P$ as the unique PTS
  $( \mkern1mu \closedSTerms, \, \Act, \, {\trans} \mkern1mu )$ such
  that 
   $t \trans[a] \pi$ iff
   $t \trans[a] \theta \in \Tr$ and $\sem{\theta}{}=\pi$ for
  some~$\theta \in \closedDTerms$. 
\end{definition}
 
\blankline

\begin{example}
  \label{ex:running:rules}
  The rules for the process algebra of Example~\ref{ex:running:sig}
  are defined by
  \begin{displaymath}
    \dedrule{}{a.\mu \trans[a] \mu}
    \qquad \qquad \qquad
    \dedrule{x\trans[a]\mu}{{x+y}\trans[a]\mu}
    \qquad \qquad \qquad
    \dedrule{y\trans[a]\mu}{{x+y}\trans[a]\mu}
    \smallskip
  \end{displaymath}%
  We use $\runningPTSS$ for the PTSS defined by these rules. 
\end{example}  

%% file: bisimulations.tex
\section{Branching bisimulations for probabilistic transition systems} 
\label{section:bisimulations}

For a set~$X$, we denote by $\SD(X)$ the set of discrete
sub-probability distributions over~$X$.
Given $\pi \in \SD(X)$, we denote by $\support(\pi)$ its support $\lc
x \in X \mid \pi(x) > 0 \rc$, and by $\pi(\bot)$ the value $1 -
\pi(X)$, for a distinguished symbol~$\bot \not \in X$.
We use $\delta_\bot$ to represent the empty distribution,
i.e.\ $\delta_\bot(X) = 0$.

An \emph{execution fragment} of a PTS~$\A$ is a finite or infinite
alternating sequence of states and actions $\alpha = s_0 \mkern2mu a_1
s_1 a_2 s_2 \ldots$ such that $s_{i-1}\trans[a_i] \pi_i$ and $s_i \in
\support(\pi_i)$, for each~$i > 0$.  We say, $\alpha$ is starting from
$\first(\alpha) = s_0$, and in case the sequence is finite, ending
in~$\last(\alpha)$. Put $\length(\alpha) = n$ if $\alpha = s_0 a_1
s_1, \ldots a_n s_n$, and $\length(\alpha) = \infty$ if $\alpha$ is
infinite. We write $\frags(\A)$ for the set of execution fragments
of~$\A$, and by $\frags^\ast(\A)$ the set of finite execution
fragments of~$\A$.  An execution fragment~$\alpha$ is a prefix of an
execution fragment~$\alpha'$, denoted by $\alpha \preccurlyeq
\alpha'$, if the sequence~$\alpha$ is a prefix of the
sequence~$\alpha'$.  The trace~$\trace(\alpha)$ of~$\alpha$ is the
subsequence of external actions of~$\alpha$. We use~$\varepsilon$ to
denote the empty trace. Similarly, we define~$\trace(a) = a$, for~$a
\in \Act$, and~$\trace(\tau) = \varepsilon$.

\newtext{
Given a state $s$, we write $\vec{s}$ 
to denote the set of outgoing transitions of the state $s$.
A \emph{scheduler} for a PTS~$\A$ is a function $\scdl :
\frags^\ast(\A) \to \SD(\trans)$ with $\scdl(\alpha) \in
\SD(\overrightarrow{\last(\alpha)})$. 
A scheduler $\scdl$ is \emph{deterministic} if for all $\alpha \in
\frags^\ast(\A)$ $\scdl(\alpha)$ is either a Dirac distribution or the
empty distribution.}
Note that by using sub-probability distributions, it is
possible that with non-zero probability no transition is chosen
after~$\alpha$, that is, the computation stops with
probability~$\scdl(\alpha)(\bot)$.  
Given a scheduler~$\scdl$ and a
finite execution fragment~$\alpha$, the distribution~$\scdl(\alpha)$
describes how transitions are chosen to move on from~$\last(\alpha)$.
A scheduler~$\scdl$ and a state~$s$ induce a probability distribution
$\pisigmas$ over execution fragments on measurable sets as generated
by the cones of finite execution fragments. The cone~$C_{\alpha}$ of a
finite fragment~$\alpha$ is the set $\lc \alpha' \in \frags(\A) \mid
\alpha \preccurlyeq \alpha' \rc$. With respect to~$\scdl$ and~$s$,
the probability~$\pisigmas$ of the cone~$C_{\alpha}$ is recursively
defined by
\begin{align*}
  \pisigmas(C_t) & = \delta_s(t) \\
  \pisigmas(C_{\alpha{at}}) & = 
  \textstyle
  \pisigmas(C_{\alpha}) \cdot \sum_{\last(\alpha)\trans[a]
    \pi} \: \scdl(\alpha)(\last(\alpha)\trans[a] \pi)\cdot \pi(t)
\end{align*}
Given a scheduler~$\scdl$, a state~$s$ and a finite execution
fragment~$\alpha$, the \emph{probability of executing~$\alpha$} based
on~$\scdl$ and~$s$, notation~$\pisigmas(\alpha)$, is defined as
$\pisigmas(\alpha) = \pisigmas(C_\alpha) \cdot \scdl(\alpha)(\bot)$.
We define the length of a scheduler~$\scdl$ with respect to a
state~$s$ by $\length_s(\scdl) = \max \lc \length(\alpha) \mid
\first(\alpha)=s ,\, \scdl(\alpha) \neq \delta_{\bot} \rc$.  Given a
scheduler~$\scdl$, we may define a truncation~$\scdl_n$ by
$\scdl_n(\alpha) = \scdl(\alpha)$ if $\length(\alpha) \leqslant n$,
otherwise $\scdl_n(\alpha) = \delta_\bot$.

We say that a state~$s$ can execute a weak transition for action~$a
\in \Act$ if there is a scheduler~$\scdl$ such that the action~$a$ is
executed with probability~$1$.
After~$a$ is executed, with probability $\pisigmas (\lc \alpha\in
\frags^\ast(\A) \mid \last(\alpha) = t \rc)$ the state~$t$ will be
reached. If~$a = \tau$, we have a similar definition but with
probability~$1$ no visible action is executed.

\blankline

\begin{definition}
  \label{def:combined-trans}
  Let~$\A$ be a PTS with $s \in S$ and $a\in \Sigma \cup\{\varepsilon\}$. A transition
  $s \transs[a]_c \pi$ is called a \emph{weak combined} transition if
  there exists a scheduler~$\scdl$ such that the following holds for
  $\pisigmas$:
  \begin{enumerate}
  \item $\pisigmas( \, \frags^\ast(\A) \, )= 1$.
  \item For each~$\alpha \in \frags^\ast(\A)$, if $\pisigmas(\alpha) >
    0$ then $\trace(\alpha) = \trace(a)$.
  \item For each state~$t$, $\pisigmas (\lc \alpha\in \frags(\A^\ast)
    \mid \last(\alpha) = t \rc) = \pi(t)$.
  \end{enumerate}
\end{definition}

\blankline

\noindent
Occasionally we want to make reference to the scheduler~$\scdl$
underlying a weak combined transition $s \transs[a]_c \pi$. We do so
by writing $s \transs[a]_\scdl \pi$.
\newtext{If the scheduler is deterministic, we may write $s \transs[a] \pi$.}

\begin{figure}
\hspace*{-0.5cm}
\begin{minipage}{.6\textwidth}
  \centering
\begin{tikzpicture}
 \node (s1) at (0,0) {$s_1$}; 
 \node (pi1) at (1.3,0) {$\pi_1$};
 \node (s2) at (2.6,1.5) {$s_2$};
 \node (s3) at (2.6,0) {$s_3$}; 
 \node (pi2) at (3.9,1.5) {$\pi_2$};
 \node (pi3) at (3.9,0) {$\pi_3$}; 
 \node (s4) at (5.2,1.5) {$s_4$};
 \node (s5) at (5.2,0) {$s_5$};
 \path[->] (s1)  edge node [above] {$\tau$} (pi1);
 \path[->] (pi1) [dotted, out=90,in=180] edge node [left] {$0.5$} (s2);
 \path[->] (pi1) [dotted] edge node [above] {$0.5$} (s3);
 \path[->] (s2)  edge node [above] {$b$} (pi2);
 \path[->] (s2)  edge node [below=0.0cm, left=0.1cm] {$a$} (pi3);
 \path[->] (s3)  edge node [above] {$a$} (pi3);
 \path[->] (pi2) [dotted] edge node [above] {$1$} (s4);
 \path[->] (pi3) [dotted] edge node [above] {$1$} (s5);
 \node (s0) at (-2.6,-1) {$s_0$};   
 \node (pi0) at (-1.3,-1) {$\pi_0$};   
 \node (s6) at (0,-2) {$s_6$};   
 \node (pi6) at (1.3,-2) {$\pi_6$};    
 \node (s7) at (2.6,-2) {$s_7$};    
 \node (pi7) at (3.9,-2) {$\pi_7$};    
 \node (s8) at (5.2,-1) {$s_8$};    
 \node (s9) at (5.2,-3.1) {$s_9$};    
 \path[->] (s0)  edge node [above] {$\tau$} (pi0);
 \path[->] (pi0) [dotted] edge node [above left] {$0.5$} (s1);
 \path[->] (pi0) [dotted] edge node [below left] {$0.5$} (s6);
 \path[->] (s6)  edge node [above] {$a$} (pi6);
 \path[->] (pi6) [dotted] edge node [above] {$1$} (s7);
 \path[->] (s7) edge node [above] {$\tau$} (pi7);
 \path[->] (pi7) [dotted] edge node [above left] {$0.5$} (s8);
 \path[->] (pi7) [dotted] edge node [below left] {$0.5$} (s9);
\end{tikzpicture}
  \captionof{figure}{PTS of Example~\ref{ex:combined_transition}}
  \label{fig:aPTS}
\end{minipage}%
\begin{minipage}{.4\textwidth}
  \centering
\begin{tikzpicture}
 \node (t0) at (0,0) {$t_0$}; 
 \node (tpi0) at (1.2,0) {$\pi_0$};
 \node (t1) at (2.4,0) {$t_1$};  
 \node (tpi1) at (3.6,1) {$\pi_1$}; 
 \node (tpi2) at (3.6,0) {$\pi_1'$};  
 \node (tpi3) at (3.6,-1) {$\pi_1''$}; 
 \node (t4) at (4.8,1) {$t_2$}; 
 \node (t5) at (4.8,-1) {$t_3$};  
 \node (tpi4) at (6,1) {$\pi_2$}; 
 \node (tpi5) at (6,-1) {$\pi_3$};  
 \path[->] (t0)  edge node [above] {${\small \tau}$} (tpi0);
 \path[->] (tpi0) [dotted] edge node [above] {$1$} (t1); 
 \path[->] (t1)  edge node [above left=-0.05cm] {$a$} (tpi1);
 \path[->] (t1)  edge node [above=-0.075cm] {$a$} (tpi2);
 \path[->] (t1)  edge node [below left=-0.05cm] {$a$} (tpi3);
 \path[->] (tpi1) [dotted] edge node [above] {$1$} (t4);
 \path[->] (tpi2) [dotted] edge node [below right] {$0.5$} (t4);
 \path[->] (tpi2) [dotted] edge node [above right] {$0.5$} (t5); 
 \path[->] (tpi3) [dotted] edge node [above] {$1$} (t5);
 \path[->] (t4)  edge node [above] {$b$} (tpi4);
 \path[->] (t5)  edge node [above] {$c$} (tpi5);
\end{tikzpicture}
\begin{tikzpicture}
 \node (t1) at (2,0) {$u_1$};  
 \node (tpi1) at (3.2,0.8) {$\pi_1$}; 
 \node (tpi3) at (3.2,-0.8) {$\pi_1''$}; 
 \node (t4) at (4.4,0.8) {$t_2$}; 
 \node (t5) at (4.4,-0.8) {$t_3$};  
 \node (tpi4) at (5.6,0.8) {$\pi_2$}; 
 \node (tpi5) at (5.6,-0.8) {$\pi_3$};  
 \path[->] (t1)  edge node [above left=-0.05cm] {$a$} (tpi1);
 \path[->] (t1)  edge node [below left=-0.05cm] {$a$} (tpi3);
 \path[->] (tpi1) [dotted] edge node [above] {$1$} (t4);
 \path[->] (tpi3) [dotted] edge node [above] {$1$} (t5);
 \path[->] (t4)  edge node [above] {$b$} (tpi4);
 \path[->] (t5)  edge node [above] {$c$} (tpi5);
\end{tikzpicture}
  \captionof{figure}{$t_0 \pbbisim u_1$ but $t_0 \nbbisim u_1$}
  \label{fig:pbbisim_vs_bbisim}
\end{minipage}
\end{figure}

\blankline

\noindent\begin{example}\label{ex:combined_transition}
\newtext{
  The state~$s_0$ in Figure~\ref{fig:aPTS}, 
  has the following weak combined transitions:
  \begin{inparaenum}[(i)] 
  \item $s_0 \transs[\varepsilon] \delta_{s_0}$ 
  \item $s_0 \transs[\varepsilon]_c \pi_{s_0}$ with 
    $\pi_{s_0}(s_0) = \pi_{s_0}(s_2) =  \pi_{s_0}(s_3) = 0.2$ and $\pi_{s_0}(s_6)=0.4$
  \item $s_0 \transs[a] \pi_a$ with $\pi_a(s_5) = \pi_a(s_7) = 0.5$
  \item $s_0 \transs[a]_c \pi'_a$ with $\pi_a'(s_5) = 0.5$, 
    $\pi_a'(s_7) = 0.2$ and $\pi_a'(s_8) = \pi_a(s_9) = 0.15$.  
 \end{inparaenum}
Notice combined transitions in $(i)$ and $(iii)$ are defined using deterministic schedulers. 
There is no weak combined transition from~$s_0$ with
action~$b$, since there is no scheduler that allows to execute~$b$ with probability~$1$.
}
\end{example}


\noindent
Next, we generalize the notion of a weak combined transition to allow
for a distribution of source states.

\blankline

\begin{definition}
  \label{def:hyper-trans}
  Let $\A = ( \mkern1mu \closedSTerms, \, \Act, \, \trans \mkern1mu )$
  be a PTS, $\pi, \pi' \in \SD(S)$ sub-probability distributions, and
  $a \in \Act$ an action. We say that $\pi \transs[a]_c \pi'$ is a
  weak combined hyper transition if there exists a family of weak
  combined transitions $\lc s \transs[a]_c \pi_s \rc_{s \in
    \support(\pi')}$ such that $\textstyle \pi' = \sum_{s \in
    \support(\pi)} \: \pi(s) \cdot \pi_s$.
\end{definition}

\blankline

\noindent
In the sequel we will consider weak combined transitions and weak
combined hyper transitions consisting of transitions taken from a
specific subset, the so-called \emph{allowed} transitions. This leads
to the notion of allowed transitions.

\blankline

\noindent\begin{definition}
  \label{def:allowed-trans}
  Choose, for a PTS $\A = ( \mkern1mu \closedSTerms, \, \Sigma, \,
  \trans \mkern1mu )$, a subset of
  transitions $P \subseteq {\trans}$. We say that $s \transs[a \rest P]_c
  \pi$ is an \emph{allowed} weak combined transition from~$s$ to~$\pi$
  respecting~$P$, if there exists a scheduler~$\scdl$ inducing a weak
  combined transition $s \transs[a]_c \mu$ such that, for each $\alpha
  \in \frags^\ast(\A)$, $\support(\scdl(\alpha)) \subseteq
  P$. Similarly, we say that $\pi \transs[a \rest P]_c \pi'$ is an
  \emph{allowed} weak combined hyper transition from~$\pi$
  to~$\pi'$ respecting~$P$, if there exists a family of allowed weak
  combined transitions $\lc s \transs[a \rest P]_c \pi_s \rc_{s \in
    \support(\pi)}$ such that $\textstyle \pi' = \sum_{s \in
    \support(\pi)} \: \pi(s) \cdot \pi_s$.
\end{definition}

\blankline

\noindent
We write $s \trans[a]_c \pi$ if there is a scheduler~$\scdl$ such
that $\length(\scdl) = 1$ and $s \transs[a]_\scdl \pi$. We write
$\pi \trans[a]_c \pi'$ if $\pi \transs[a]_c \pi'$ with an
associated family of transitions $\lc s \trans[a]_c \pi_s \rc_{s \in
\support(\pi)}$. 
\newtext{We employ similar notation for combined transitions with 
deterministic schedulers and allowed transitions.}

The paper~\cite{EHKTZ2013} proposes to consider in their study of
weak probabilistic bisimulation the notion of an allowed
transition (see Definition~\ref{def:allowed-trans}). The key idea is
to consider transitions taken from a specific part of the transition
system. 
\newtext{We use the essence of this to distill the notions of
\emph{probabilistic branching bisimulation} and 
\emph{branching bisimulation}}.

Given a relation ${\B} \subseteq \closedSTerms \times
\closedSTerms$, its lifting to $\Delta(\closedSTerms) \times
\Delta(\closedSTerms)$ is defined as follows:
$\pi \B \pi'$ iff there is a weight function $w:(\closedSTerms \times
\closedSTerms) \to [0,1]$ such that for all $t,t' \in \closedSTerms$,
\begin{inparaenum}[(i)]
\item\label{eq:weigh-left}%
  $w(t,\closedSTerms) = \pi(t)$,
\item\label{eq:weigh-right}%
  $w(\closedSTerms,t') = \pi'(t')$, and
\item\label{eq:weigh-rel}
  $w(t,t') > 0 \text{ implies } t \B t'$.
\end{inparaenum}
It is easy to check that the weight function is a probability
distribution on $\closedSTerms \times \closedSTerms$.
Moreover, the lifting of~$\B$ is reflexive, symmetric and/or
transitive if $\B$~is.  The overloading of~$\B$
should not raise confusion.

\begin{definition}
  \label{def:bbisim}
  Let $\A = ( \mkern1mu \closedSTerms, \, \Act, \, {\trans} \mkern1mu
  )$ be a PTS, and let $\B \subseteq \closedSTerms \times
  \closedSTerms$ be a symmetric relation. A transition $s \trans[\tau]
  \pi$ is called \emph{branching preserving} for~$\B$ if $\delta(s) \B
  \pi$.
  The relation~$\B$ on~$\closedSTerms$ is called a 
  \emph{probabilistic branching bisimulation} 
  \newtext{(resp. \emph{branching bisimulation})}
  for~$\A$ if, for a branching preserving set of transitions $P \subseteq {\trans[\tau]}$, 
  it holds that $s \B t$ and $s \trans[a] \pi_s$ imply either 
  \begin{enumerate}[(i)]
   \item $a = \tau$ and $s \trans[\tau] \pi_s \in P$, or
   \item $t \transs[\tau \rest P]_c \tilde{\pi}_t$ with $\tilde{\pi}_t \trans[a]_c \pi_t$ 
   \newtext{(resp. $\tilde{\pi}_t \trans[a] \pi_t$)} and 
   $\pi_s \B \pi_t$, for all states $s, t \in \closedSTerms$.
  \end{enumerate}
   \newtext{
   We write $s \pbbisim t$ (resp. $s \bbisim t$) 
   if there exists a probabilistic branching bisimulation 
   (resp. branching bisimulation) for~$\A$ relating states $s$ and~$t$}.
\end{definition}

\begin{example}
 \newtext{
 In Figure~\ref{fig:pbbisim_vs_bbisim} we can see the subtle difference between
 the relations ${\pbbisim}$ and ${\bbisim}$.
 Comparing the transitions of $t_1$ with the transitions of $u_1$,
 we have that the transition $t_1 \trans[a] [0.5] t_2 \oplus [0.5] t_3$
 can be mimicked by $u_1$ with the transitions 
 $u_1 \trans[\tau] \delta_{u_1}$ and $\delta_{u_1} \trans[a]_c [0.5]t_2 \oplus [0.5] t_3$.
 Therefore it is clear that $t_1 \pbbisim u_1$ and $t_0 \pbbisim u_1$. In this case, $t_0 \trans[\tau] \delta_{t_1}$ 
 is branching preserving.
 On the other hand there is no transition $\delta_{u_1} \trans[a] [0.5]t_2 \oplus [0.5] t_3$, 
 therefore $t_1 \nbbisim u_1$ and $t_0 \nbbisim u_1$
 }
\end{example}

\noindent
In~\cite{AGT2012}, a definition for branching bisimulation for the
alternating model~\cite{Han91:phd} is presented.
Contrary to our definition above, it does not use the notion of
schedulers.
We explain the underlying intuition considering the PTS depicted in
Figure~\ref{fig:no_schedulers}. 
It is clear that states $T = \lc t_1, \ldots, t_4 \rc$ are branching
bisimilar and that all transitions labeled~$\tau$ are branching
preserving.
Notice that for every scheduler~$\scdl$ such that $s\transs[a]_\scdl
\pi_\scdl$ we have~$\pi_\scdl = \pi$.
Thus, in this particular example, schedulers do not make any
difference for the definition of a combined transition executing~$a$.
In general, it does not make any change choosing, with different
probabilities, between branching preserving transitions until a
visible action is executed, because the targets of two branching
preserving transitions are, by definition, the same modulo branching
bisimulation.

\begin{figure}
\begin{center}
\begin{tikzpicture}
 \node (s0) at (0,0) {$s_0$}; 
 \node (pi0) at (2,0) {$\pi_0$};
 \node (s1) at (4,2) {$s_1$}; 
 \node (s2) at (4,0) {$s_2$};  
 \node (s3) at (4,-2) {$s_3$}; 
 \node (pi1) at (6,2) {$\pi_1$}; 
 \node (pi2) at (6,0) {$\pi_2$};  
 \node (pi3) at (6,-2) {$\pi_3$}; 
 \node (t1) at (8,2) {$t_1$}; 
 \node (t2) at (8,1) {$t_2$};  
 \node (t3) at (8,-1) {$t_3$}; 
 \node (t4) at (8,-2) {$t_4$};  
 \node (pi) at (10,0) {$\pi$}; 
  \path[->] (s0)  edge node [above] {$\tau$} (pi0);
  \path[->] (pi0) [dotted] edge node [above left] {$p_1$} (s1); 
  \path[->] (pi0) [dotted] edge node [above] {$p_2$} (s2); 
  \path[->] (pi0) [dotted] edge node [below left] {$p_3$} (s3);
  \path[->] (s1)  edge node [above] {$\tau$} (pi1);
  \path[->] (s2)  edge node [above] {$\tau$} (pi2);
  \path[->] (s3)  edge node [above] {$\tau$} (pi3);
  \path[->] (s2)  edge node [above] {$\tau$} (pi3);  
  \path[->] (pi2) [dotted, out=90,in=340] edge node [above right] {$p_4$} (s1);
  \path[->] (pi2) [dotted, out=110,in=70] edge node [above] {$p_5$} (s2);  
  \path[->] (pi1) [dotted] edge node [above] {$p_6$} (t1); 
  \path[->] (pi1) [dotted] edge node [above] {$p_7$} (t2); 
  \path[->] (pi2) [dotted] edge node [above] {$p_8$} (t2);
  \path[->] (pi2) [dotted] edge node [above] {$p_9$} (t3); 
  \path[->] (pi2) [dotted] edge node [below] {$p_{10}$} (t4); 
  \path[->] (pi3) [dotted] edge node [below] {$p_{11}$} (t4);  
  \path[->] (t1)  edge node [above] {$a$} (pi);
  \path[->] (t2)  edge node [above] {$a$} (pi);
  \path[->] (t3)  edge node [above] {$a$} (pi);
  \path[->] (t4)  edge node [above] {$a$} (pi);
\end{tikzpicture}
\vspace*{-1.5\baselineskip}
\end{center}
\caption{Schedulers are not needed to define branching bisimulation}
\label{fig:no_schedulers}
\end{figure}
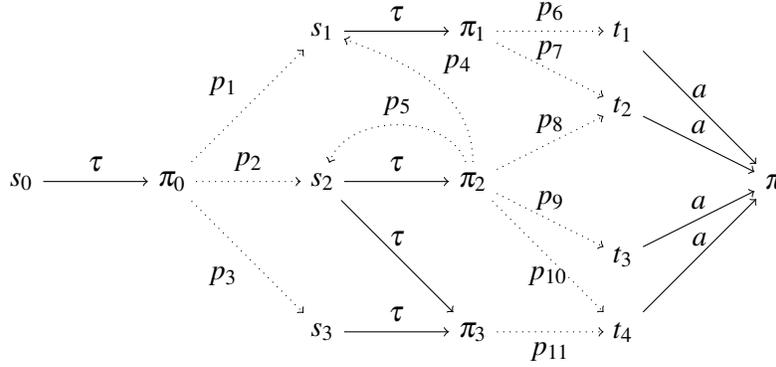

We adapt the definition of~\cite{AGT2012} for branching bisimulation
without scheduler for alternating systems to the present setting. 
This new definition coincides with Definition~\ref{def:bbisim}.
A \emph{concrete} execution of a PTS~$\A$ is a sequence $s_0 a_0 \pi_0
s_1 
\ldots \pi_{n-1} s_n a_n \pi_n$ such that $s_i
\trans[a_{i}] \pi_{i}$ for $0 \leqslant i \leqslant n$ and $s_{j+1}
\in \support(\pi_j)$ for $0 \leqslant j < n$.

\begin{definition}
  \label{def:bb-without-scheduler}
  A symmetric relation $\B \subseteq \closedSTerms \times
  \closedSTerms$ is called a \emph{branching bisimulation without scheduler}
  for a PTS $\A = ( \mkern1mu \closedSTerms, \, \Act, \, {\trans}
  \mkern1mu )$ iff $s \B s'$ and $s \trans[a] \pi$ imply (i)~$a =
  \tau$ and $\delta_{s} \B \pi$, or~(ii) $a \in \Sigma$ and there is a
  concrete execution $s_0 \tau \pi_1 s_1 \tau \pi_2 s_2 \ldots \pi_{n}
  s_n a \pi'$ such that $s' = s_0$, $s \B s_i$ for $0 \leqslant i
  \leqslant n$, $\delta_s \B \pi_j$ for $1 \leqslant j \leqslant n$
  and $\pi \B \pi'$.
  We write $s \nsbisim s'$ if a branching bisimulation without
  scheduler for~$\A$ relates $s$ and~$s'$. 
\end{definition}

\noindent
Thus, while the scheduler-based definition requires preservation of
branching by the transitions that are combined, the scheduler-less
definition requires branching bisimilarity of the intermediate states
and distribution with the source state. In view of this, the next
result does not come as a surprise.

\begin{theorem}
  \label{th:with-and-without-are-eq}
  Let $\A = ( \mkern1mu \closedSTerms, \, \Act, \,  {\trans} \mkern1mu
  )$ be a PTS. Then $s \bbisim t$ iff $s
  \nsbisim t$, for all $s,t \in \closedSTerms$.
  \qed
\end{theorem}

\blankline

\noindent
It is well-known that typically branching bisimulation per se is not
preserved by most process algebras~\cite{Fok00,CLT14:fi}. 
This is solved strengthening the requirements by means of the
\emph{rootedness condition}.  In this way, we obtain a stronger
notion, in our setting that of a rooted branching
bisimulation.

\begin{definition}
  A symmetric relation ${\rel} \subseteq \closedSTerms \times
  \closedSTerms$ is a \emph{rooted branching bisimulation} with
  respect to a PTS $\A = ( \mkern1mu \closedSTerms, \, \Act , \,
  {\trans} \mkern1mu )$ if, for all $s,t \in \closedSTerms$,$s \rel t$
  and $s \trans[a] \theta_s$ imply $t\trans[a] \theta_t$ for some
  $\theta_t$ s.t. $\theta_s \bbisim \theta_t$.
\end{definition}

%% file: format.tex
\section{The probabilistic RBB format}
\label{sec:format}

A congruence over an algebraic structure is an equivalence relation on
the elements of the algebra compatible with its structure.
Formally, a (sort-respecting) equivalence relation
${\relR}\subseteq\closedTerms\times\closedTerms$ is a
\emph{congruence} if
for all $f\in\Sigma$ and $\xi_i,\xi'_i \in \closedTerms$ with $\xi_i
\relR \xi_i'$ for all~$i$, $1 \leqslant i \leqslant \rank(f)\}$,
it holds that $f(\xi_1,\ldots,\xi_{\rank(f)}) \relR
f(\xi'_1,\ldots,\xi'_{\rank(f)})$. 
%

We next explain the restrictions needed by the format.
Then we present the probabilistic RBB safe specification format.
If a PTSS~$P$ satisfies this specification format then rooted
branching bisimulation is a congruence for all operator defined
by~$P$.
Finally, we present the definitions and lemmas needed to prove the
main result.
Basically, we have extended the ideas from \cite{Fok00} to our
setting.



\myparagraph{Restrictions over the format.}
Consider, in the context of our running example, the terms $s =
a.(\dist{b.\nullproc})$ and $t = a.(\dist{\tau.(b.\nullproc)})$.
Note~$s \rbbisim t$.
The restrictions on the format discussed below are justified by the
following examples.

\emph{Look-ahead is not allowed.}
The format presented in \cite{DL12,LGD12} allows \emph{quantitative
  premises}.
These premises have the shape $\theta(Y) \gtgeq p$ with 
$\theta \in \openDTerms,\ Y \subseteq \TVar,\ \gtgeq \in \{>, \geq\}$
and $p\in[0,1]$. 
Given a closed substitution $\rho$, the closed premise $\rho(\theta(Y)
\gtgeq p)$ holds 
iff $\sem{\rho(\theta)}{}(\rho(Y)) \gtgeq p$ holds.
Without quantitative premises it is not possible to have look-ahead.
Suppose that our rules (Definition~\ref{def:ptss}) support this kind
of premises, then we can add to $\runningPTSS$ the operator~$f$,
$\arity(f) = \TS \to \TS$, defined by
\begin{equation}\label{rules:lookahead}
  \dedrule{x \trans[a] \mu \quad \mu(\{y\}) > 0  \quad y \trans[b]
    \nu}{f(x) \trans[a] \dist{\nullproc}} 
\end{equation}
Then $f(s) \trans[a]$ and $f(t) \ntrans[a]$, therefore~$f(s) \notrbbisim f(t)$. 
The problem with look-ahead is the possibility for testing after an
action is executed.
The processes reached (with probability greater than zero) after the
execution of the action may not be rooted branching bisimilar. 
To avoid this problem and simplify the presentation, the kind of rules
that we are using does not support quantitative premises.

The following examples follow in essence the same idea: 
by combining operators with particular rules, it is possible to test,
in some way, the target of a positive premise.
In addition, we use the examples to motivate other ingredients of our
format below.

\emph{Arbitrary testing of a positive premise target using other rules
  is not allowed}.
Add to $\runningPTSS$ the operator~$f$ with $\arity(f) = \TS \to \TS$,
and operator~$g$ with $\arity(g) = \TS\TS \to \TS$
defined by 
\begin{equation} \label{rules:f_g_base}
  \dedrule{x \trans[a] \mu}{f(x) \trans[a] \dist{g}(\delta(x), \mu)}
  \qquad \qquad
  \dedrule{x_2 \trans[b] \mu }{g(x_1, x_2) \trans[b] \dist{\nullproc}}
\end{equation}
then $f(s) \notrbbisim f(t)$.
Notice, $f(s) \trans[a] \dist{g}(\delta(s), \dist{b.\nullproc})$ and
$g(s, b.\dist{\nullproc}) \trans[b] \dist{\nullproc}$, while $f(t)
\trans[a] \dist{g}(\delta(t), \dist{\tau.b.\nullproc})$ and $g(t,
\tau.b.\dist{\nullproc}) \ntrans[b]$.
Here, the target of the premise $x \trans[a] \mu$ of the first rule is
tested, via instantiation of~$x_2$, in the second rule.
When the target of a positive premise is used as an argument of the
function in the target of the conclusion, we say that the position of
the argument is \emph{wild}. If the position is not wild, then it is
\emph{tame}.
Here, the second argument of~$g$ is wild.
To `control' wild arguments, following~\cite{Fok00}, \emph{patience
  rules} are introduced, see Definition~\ref{def:patience_rule}.
In this case, a patience rule for the second argument of~$g$ is defined by
\begin{equation}\label{rules:f_g_patience}
  \dedrule{x_2 \trans[\tau] \mu}{g(x_1, x_2) \trans[\tau]
    \dist{g}(\delta(x_1), \mu)} 
  \vspace{3pt}
\end{equation}
Taking into account this new rule we have that~$f(s) \bbisim f(t)$.

However, the problem we have presented is more general. Add to
$\runningPTSS$ operator~$f$ with $\arity(f) = \TS \to \TS$,
operator~$g$ with $\arity(g) = \TS\TS \to \TS$ and operator~$h$ with
$\arity(h) = \TS \to \TS$ defined by
\begin{equation} \label{rules:f_g_h}
\dedrule{x \trans[a] \mu}{f(x) \trans[a] \dist{h}(\dist{g}(\delta(x), \mu))}
\qquad\qquad
\dedrule{x_2 \trans[b] \mu }{g(x_1, x_2) \trans[b] \dist{\nullproc}}
\qquad\qquad
\dedrule{x_1 \trans[b] \mu }{h(x_1) \trans[b] \dist{\nullproc}}
  \vspace{3pt}
\end{equation}
Then~$f(s) \notrbbisim f(t)$. 
In the target of the conclusion of the first rule, the second argument
of $\dist{g}$ is the target of a positive premise.  In addition, this
term is the argument of~$\dist{h}$.
In this case, we have that the second argument of~$g$ and the argument
of~$h$ are wild.
Also here, terms $f(s)$ and~$f(t)$ can become rooted branching
bisimilar adding patience rules for the wild arguments.
In the next section, Definition~\ref{def:wild_argument} formalizes the
notion of a wild argument and Definition~\ref{def:w-nested} is
introduced to deal with their nesting.

Finally, we point out that the condition of being wild can be
\emph{`inherited'}.  For example, take into account rules from
(\ref{rules:f_g_base}) and (\ref{rules:f_g_patience}).  Recall, the
second argument of~$g$ is wild and its patience rule is defined.
Add a new operator~$h$ with $\arity(h) = \TS\TS \to \TS$ and the
following rules
\begin{equation}\label{rules:inherited}
  \dedrule{}{g(x_1,x_2) \trans[a] \delta(h(x_2,x_1))} 
  \qquad \qquad
  \dedrule{x_1 \trans[b] \mu}{h(x_1,x_2) \trans[c] \dist{0}} 
  \medskip
\end{equation}
the first argument of~$h$ is wild because variable~$x_2$ appears in
that position and it also appears in a wild position of~$g$.  If we do
not add a patience rule for the first argument of~$h$ we have~$f(s)
\notrbbisim f(t)$.

\emph{Wild arguments cannot be used in the source of a premise
  unrestrictedly}.
Again, take into account rules from (\ref{rules:f_g_base}) and
(\ref{rules:f_g_patience}) and add either one of the following rules
\begin{equation}
  \dedrule{x_2 \ntrans[b]}{g(x_1,x_2) \trans[a] \dist{\nullproc}} 
  \qquad \qquad
  \dedrule{x_2 \trans[\tau] \mu}{g(x_1,x_2) \trans[a]
    \dist{\nullproc}} 
  \medskip
\end{equation}
In both cases we get~$f(s) \notrbbisim f(t)$.
The first rule is an example that shows that a wild argument cannot
appear as the source of a negative premise.
The second rule shows that a wild argument cannot be used as the
source of a positive premise with label~$\tau$.
Both kind of restrictions will be present in the format.

\myparagraph{The RBB format.}
Definition~\ref{def:ngraph} introduces the \emph{nesting graph}: an
edge from~$\tuple{f,i \mkern1mu }$ to~$\tuple{g,j \mkern1mu }$ in the
graph encodes that a variable appearing in the $i$-th argument of~$f$,
also appears in a term that is used as the $j$-th argument of~$g$ in
the conclusion of a rule.
This graph is used to define when a variable is~wild.

\blankline

\begin{definition}
  \label{def:ngraph}
  \begin{itemize}
  \item [(a)]
  Let $P = ( \mkern1mu \Sigma,\, \Act, \, \calR \mkern1mu )$ be a
  PTSS\@.  The \emph{nesting graph} of the
  PTTS~$P$ 
  is the directed graph $\calG_P = (V,E)$ with
  \begin{eqnarray*}
    V & = &
    \lc \tuple{f,i} \mid 
    f \in \Sigma_\TS, 1 \leqslant i \leqslant \rank(f) \rc
    \\
    E & = &
    \lc (\tuple{f,i}, \tuple{g,j}) \mid 
    r \in \calR, \,
    \conc{r} = f( \, .. , \zeta_{i} , .. \mkern1mu ) \trans[a]
    C[g( \, .. , , \xi_{j-1}, C'[\zeta_i], \xi_{j+1}, , .. \mkern1mu )] 
    \rc \cup {} \\ 
    & &
    \lc (\tuple{f,i}, \tuple{g,j}) \mid 
    r \in \calR, \,
    \conc{r} = f( \, .. , \zeta_{i} , .. \mkern1mu ) \trans[a]
    C[\dist{g}( \, .. , , \xi_{j-1}, C'[\zeta_i], \xi_{j+1}, , .. \mkern1mu )] 
    \rc
\end{eqnarray*}
  \item [(b)]
    \label{def:wild_argument}
    Let $P$ be a PTSS and $\calG_P = (V,E)$ be its nesting graph. A
    node $\tuple{g,j} \in V$ is called \emph{wild} if
    \begin{enumerate}
    \item \label{def:wild:basecase_s}
      $t \trans[a] \mu \in \prem{r}$ and
      $\trgt{\conc{r}} = C[g( \, .. , \xi_{j-1}, C'[\mu],
      \xi_{j+1}, .. \mkern1mu )]$ 
      for some~$r\in \calR$, or
    \item \label{def:wild:basecase_d}
      $t \trans[a] \mu \in \prem{r}$ and
      $\trgt{\conc{r}} = C[\dist{g}( \, .. , \xi_{j-1}, C'[\mu],
      \xi_{j+1}, .. \mkern1mu )]$ 
      for some~$r\in \calR$, or
    \item 
      \label{def:wild:indcase} 
      $(\tuple{f,i}, \tuple{g,j}) \in E$
      and $\tuple{f,i}$ is wild.
   \end{enumerate}
   The $i$-th argument of an operator~$f$ or~$\dist{f}$ is \emph{wild}
   if $\tuple{f,i}$ is wild, otherwise it is \emph{tame}.
  \end{itemize}
\end{definition}

\noindent
In the definition of a wild argument, items
(b.\ref{def:wild:basecase_s}) and (b.\ref{def:wild:basecase_d}) deal
with the case where the target of a positive premise $\mu$ is used in
the target of the conclusion.
Two cases are needed to deal with $g \in \Ssignature$ and its lifting
$\dist{g} \in \Dsignature$. 
Notice that $g \in \Ssignature$ can appear in a distribution term,
particularly in the target of the conclusion, using the operator
$\delta$.
The nesting of context allows to take into account all possible cases. 
Item (b.\ref{def:wild:indcase}) deals with the `inheritance' mentioned above.
Below, the definitions of \emph{patience rules}, \emph{w-nested
  contexts} and \emph{the RBB safe specification format} are given.
Later we explain how these are used to ensure that rooted branching
bisimulation is a congruence.

\begin{definition}
  \label{def:patience_rule}
  Let $\zeta \in \TVar \cup \DVar$, define $\overline{\zeta} =
  \delta(\zeta)$ whenever $\zeta$ has sort~$\TS$ and $\overline{\zeta} =
  \zeta$, otherwise.
  The \emph{patience rule} for the $i$-th argument (of sort $\TS$) of
  a function symbol $f \in \Sigma_\TS$ is the rule
  \begin{displaymath}
    \dedrule
      {x_i \trans[\tau] \mu}
      {f( \mkern1mu .. , \zeta_{i-1}, x_i, \zeta_{i+1}, .. \mkern1mu ) 
        \trans[\tau]  
        \dist{f}( \mkern1mu .. , \overline{\zeta}_{i-1}, \mu,
        \overline{\zeta}_{i+1}, .. \mkern1mu )}
  \end{displaymath}
\end{definition}

\blankline

\noindent

\begin{definition}\label{def:w-nested}
  The \emph{set of w-nested contexts} is defined inductively by
  \begin{enumerate}
  \item The empty context $[\:]$ is w-nested. 
  \item The term $f( \mkern1mu .. , \xi_{i-1}, \, C[\:], \, \xi_{i+1},
    .. \mkern1mu )$ is w-nested if $C[\:]$ is w-nested and the $i$-th
    argument of the function symbol~$f$ is wild.
  \item The term $\delta(C[\:])$ is w-nested if $C[\:]$ is w-nested.
  \item The term $\bigoplus_{i\in I} \: [p_i] \mkern1mu \xi_i$ is
    w-nested if $\xi_j = C[\:]$ for some $j \in I$, and $C[]$ is
    w-nested.
  \item The term $\dist{f}( \mkern1mu .. , \xi_{i-1}, \, C[\:], \,
    \xi_{i+1}, .. \mkern1mu )$ is w-nested if $C[\:]$ is w-nested and
    the $i$-th~argument of the function symbol~$\dist{f}$ is wild.
  \end{enumerate}
  The variable~$\zeta$ \emph{appears in a w-nested position} in $\xi
  \in \openTerms$ if there is a w-nested context~$C[\:]$ with
  $C[\zeta] = \xi$.
\end{definition}


\begin{definition}
  \label{def:format}
  Let $P = ( \mkern1mu \Sigma, \, \Act , \, \calR \mkern1mu )$ be a
  PTSS where each argument of~$f \in \Sigma$ is defined as wild or
  tame with respect to~$\calG_P$.
  The PTSS~$P$ is in the \emph{probabilistic RBB safe} specification
  format if for all $r \in \calR$ one of the following conditions
  holds.
  \begin{enumerate}
  \item $r$ is a patience rule for a wild argument of a function
    symbol in $\Sigma$. 
  \item $r$ is a \emph{RBB safe rule}, i.e.\ $r$ has the following
    shape
    \begin{displaymath}
      \dedrule{%
        \textstyle
       \lc t_m \trans[a_m]\mu_m \mid m \in M \rc \qquad
       \lc t_n \ntrans[b_n] \mid n\in N \} 
    }{
        f( \mkern1mu \zeta_1, \ldots, \zeta_{\rank(f)} \mkern1mu )
        \trans[a] \theta } 
    \end{displaymath}
    where $M$ and~$N$ are index sets, $\zeta_k$, and $\mu_m$, with $1
    \leqslant k \leqslant \rank(f)$ and~$m \in M$, are all different
    variables, $f \in \Sigma_{\TS}$, $t_m, t_n \in \openSTerms$ and
    $\theta \in \openDTerms$, and the following conditions are met.
    \begin{enumerate}    
    \item \label{def:format:allowed_test} 
      If the $i$-th argument of~$f$ is wild and has a patience
      rule in~$\calR$, then for all $\psi \in \prem{r}$ such that
      $\zeta_i = x_i \in \Var(\psi)$, $\psi = x_i \trans[a_i] \mu_i$
      with $a_i \neq \tau$.
    \item \label{def:format:not_allowed_test} 
      If the $i$-th argument of~$f$ is wild and does not have a
      patience rule in $\calR$, then $\zeta_i$ does not occur in the
      source of a premise of~$r$.
    \item \label{def:format:w-nested_position} 
      Variables $\mu_m$, for $m \in M$, and variables $\zeta_i$,
      where $i$-th argument of~$f$ is wild, 
      only occur at w-nested positions in~$\theta$.
    \item \label{def:format:not_lookahead}  
      $\mu_m \not \in \Var(t_{m'})$ for all $m, m' \in M$.      
    \end{enumerate}
  \end{enumerate}
\end{definition}

\blankline

\noindent
Patience rules for wild arguments allow to progress along the $i$-th parameter
using the transition $x_i \trans[\tau] \mu$.
Because of the patience rules and the restriction to test a wild argument
(Definition~\ref{def:format}.\ref{def:format:allowed_test}),
$\tau$-transitions will not be detected
by the premises of the rules of a PTSS satisfying the format.
If an argument of some $f$ is wild and it has no patience rule, then it cannot be tested 
(Definition~\ref{def:format}.\ref{def:format:not_allowed_test}).
Notice that the patience rules lift the $\tau$-transition that can be executed
by the $i$-th parameter of $f$ to $f$, i.e.
if the $i$-th parameter of $f$ can execute a $\tau$-transition then
$f$ can execute a $\tau$-transition affecting only the this parameter.
Then, the lifting of a $\tau$-transition has to be also `controlled'.
This is done requiring Definition~\ref{def:format}.\ref{def:format:w-nested_position}. 
Finally, restriction Definition~\ref{def:format}.\ref{def:format:not_lookahead} is added
to avoid dealing with non-well-founded set of premises (see \cite{Gro93}). 
Because rules do not support quantitative premises, see
Definition~\ref{def:ptss}, this is not a significant restriction.

For the probabilistic RBB safe specification format we have the
following congruence result. 

\begin{theorem}
  \label{th:congruence}
  Let $P$ be a complete PTSS in RBB safe specification format.  Then
  $\rbbisim$ is a congruence relation for all operators defined in
  $P$.
\end{theorem}

\noindent
In the rest of this section we present the key elements needed to
prove Theorem~\ref{th:congruence}.
For the sake of clarity, we fix $P = ( \mkern1mu \Sigma, \, \Act, \,
\calR \mkern1mu )$, a complete PTSS\@.
For any ordinal~$\alpha$, let $\tuple{\CerTr_\alpha,\PosTr_\alpha}$ be
the 3-valued model of~$P$ constructed inductively as in
Lemma~\ref{lem:inductive-construction-of-3v-least-model} with
$\tuple{\CerTr_\lambda,\PosTr_\lambda}$ being the 3-valued least
stable model. Thus, $\CerTr_\lambda = \PosTr_\lambda$ is the
transition relation associated to~$P$.

First we introduce the relations $\relR$ and $\relB$. 
$\relR$~is the congruence closure of rooted branching bisimulation, and
$\relB$ is a kind of congruence closure of the branching bisimulation
based on wild and tame arguments and relation~$\relR$.
Notice the correlation between $\relB$ and the restrictions imposed to the 
target of the conclusion of a rule in RBB safe format.

\begin{definition}
  \label{def:relR}
  \begin{itemize}
  \item [(a)]
  The \emph{relation ${\relR} \subseteq {\closedTerms \times
      \closedTerms}$} is the smallest relation such that
  \begin{enumerate}
  \item \label{def:relR:state_base_case}
    ${\rbbisim} \subseteq {\relR}$,
  \item \label{def:relR:dist_base_case} $\dist{c} \relR \dist{c'}$
    whenever $c,c' \in \Ssignature$, $\arity (c), \, \arity (c') =
    \TS$ and $c \relR c'$, and
  \item\label{def:relR:ind_case}
    $f(\xi_1,\cdots,\xi_{\rank(f)})
    \relR f(\xi'_1,\cdots,\xi'_{\rank(f)})$ whenever $\xi_i \relR
    \xi'_i$ for all~$i$, $1 \leqslant i \leqslant \rank(f)$ and $f \in
    \Ssignature \cup \Dsignature$.
  \end{enumerate}
  \item [(b)]
    \label{def:relB}
    The relation ${\relB} \subseteq {\closedTerms \times \closedTerms}$
    is the smallest relation such that
    \begin{enumerate}
    \item \label{def:relB:state_base_case}
      ${\bbisim} \subseteq {\relB}$,
    \item \label{def:relB:dist_base_case}
      $\dist{c} \relB \dist{c'}$ whenever $c,c' \in \Ssignature$,
      $\arity (c) ,\, \arity (c') = \TS$ and $c \relB c'$, and
   \item\label{def:relB:ind_case}
    $f(\xi_1,\cdots,\xi_{\rank(f)}) \relR f(\xi'_1,\cdots,\xi'_{\rank(f)})$
    whenever $f \in \Ssignature \cup \Dsignature$, 
    \begin{itemize}
     \item $\xi_i \relR \xi'_i$ for the $i$-th tame arguments of~$f$, and
     \item $\xi_i \relB \xi'_i$ for the $i$-th wild arguments of~$f$.
    \end{itemize}
  \end{enumerate}
  \end{itemize}
\end{definition}

\blankline

\noindent
The two cases of the following lemma can be proved using structural
induction with respect to $\vartheta \in \openTerms$ and the
definitions of $\relR$ and~$\relB$.

\begin{lemma}
  \label{lemma:cong_closure_term_contexts}
  Let $\rho, \rho'$ be two closed substitutions and $\vartheta \in \openTerms$.
  \begin{itemize}
  \item If $\rho(\zeta) \relR \rho'(\zeta)$ for all $\zeta \in \TVar
    \cup \DVar$, then $\rho(\vartheta) \relR \rho'(\vartheta)$.
  \item If, for each variable $\zeta \in \Var(\vartheta)$, either
    \begin{itemize}
    \item $\rho(\zeta) \relR \rho'(\zeta)$, or
    \item $\rho(\zeta) \relB \rho'(\zeta)$ and $x$~only occurs at
      w-nested positions in $\vartheta$, 
    \end{itemize}
    then $\sigma(\vartheta) \relB \sigma'(\vartheta)$.
  \end{itemize}
\end{lemma}

\blankline

\noindent
The relations ${\relR}$ and ${\relB}$ lift properly to distributions.
This ensures that $\sem{\theta}{} \relR \sem{\theta'}{}$ and
$\sem{\theta}{} \relB \sem{\theta'}{}$ whenever $\theta \relR \theta'$
and $\theta \relB \theta'$, respectively.
Therefore we can work at the symbolic level.

\blankline

\begin{lemma}
  \label{lemma:congruence_closure_on_states_coincides_with_closure_on_distributions}
  Let $\theta, \theta'\in\closedDTerms$. If $\theta \relR \theta'$
  then $\sem{\theta}{} \relR \sem{\theta'}{}$. If $\theta \relB
  \theta'$ then $\sem{\theta}{} \relB \sem{\theta'}{}$.
\end{lemma}

\blankline

\noindent
Theorem~\ref{th:congruence} is a straight consequence of
Lemma~\ref{lemma:relR}.  Notice that clause (III$_\alpha$) uses the
characterization of branching bisimulation without schedulers.

\blankline

\begin{lemma}
  \label{lemma:relR}
  If $P$ is in probabilistic RBB safe format and $s, t \in
  \closedSTerms$, then
  \begin{enumerate}[($I_{\alpha}$)]
  \item If $s \relR t$ and $s \trans[a] \theta_s \in \CerTr_\lambda$,
    then $t \trans[a] \theta_t \in \PosTr_\alpha$ for some $\theta_t
    \in \closedDTerms$.
  \item If $s \relR t$ and $s \trans[a] \theta_s \in \CerTr_\alpha$,
    then $t \trans[a] \theta_t \in \CerTr_\lambda$ for some $\theta_t
    \in \closedDTerms$ with $\theta_s \relB \theta_t$.
  \item     
   If $s \relB t$ and $s \trans[a] \theta_s \in \CerTr_\alpha$ then either:    
   \begin{itemize}
   \item $a = \tau$ and $\delta(s) \relB \theta_s$ or
   \item there is a concrete execution $t_0 \tau \theta_1 t_1 \tau
     \theta_2 t_2 \ldots \theta_{n} t_n a \theta_t$ in
     $\CerTr_\lambda$ such that $t = t_0$, $s \relB t_i$ for $0 \leq i
     \leq n$, $\delta(s) \B \theta_j$ for $1 \leq j \leq n$ and
     $\theta_s\relB \theta_t$.
  \end{itemize}  
  \end{enumerate} 
\end{lemma} 

\begin{proof}[Proof of Theorem~\ref{th:congruence}]
  Clause~(III$_\alpha$) of Lemma~\ref{lemma:relR} yields that $\relB$
  is a branching bisimulation. Together with clause~(II$_\alpha$) this
  implies that $\relR$ is a rooted branching bisimulation.
 By definition of ${\rbbisim}$, ${\relR} \subseteq {\rbbisim}$ and 
 by definition of ${\relR}$, ${\relR} \supseteq {\rbbisim}$,
 therefore ${\relR} = {\rbbisim}$.
 Finally, by Definition~\ref{def:relR}.a.\ref{def:relR:ind_case},
 rooted branching bisimulation is a congruence.
\end{proof}

\begin{example}
  \newtext{ The following example shows that the format does not
    preserve probabilistic branching bisimulation.  Let $t_1 =
    a.\dist{b} + a.\dist{c}$ and $t_2 = t_1 +
    a.([0.5] \mkern1mu \dist{b} \oplus [0.5] \mkern1mu
      \dist{c})$ then $t_1 \pbbisim t_2$
    (we omitt the distribution $\dist{\nullproc}$).
    Consider the following
    rules:
    \begin{displaymath}
      \dedrule{x \trans[a] \mu}{f(x) \trans[a] \dist{g}(\mu,\mu)}
      \qquad 
      \dedrule{x_1 \trans[b] \mu_1 \qquad x_2 \trans[c]
        \mu_2}{g(x_1,x_2) \trans[a] \dist{\nullproc}} 
      \qquad 
      \dedrule{x \trans[\tau] \mu}{g(x,y) \trans[\tau]
        \dist{g}(\mu,\delta(y))} 
      \qquad 
      \dedrule{y \trans[\tau] \mu}{g(x,y) \trans[\tau]
        \dist{g}(\delta(x),\mu)}
      \smallskip
    \end{displaymath}
    then $f(t_1) = a.(\dist{g(b, b)}) + a.(\dist{g(c,c)})$, 
   $f(t_2) =  f(t_1) + a.(\dist{g}([0.5] \mkern1mu \dist{b} \oplus [0.5] \mkern1mu \dist{c},
      [0.5] \mkern1mu \dist{b} \oplus [0.5] \mkern1mu \dist{c})$.
    Notice it is not possible to combine distributions $\dist{g(b,
      b)}$ and $\dist{g(c,c)}$ to get the distribution $\dist{g}([0.5]
      \mkern1mu \dist{b} \oplus [0.5] \mkern1mu \dist{c}, [0.5] \mkern1mu \dist{b} \oplus
      [0.5] \mkern1mu \dist{c})$.
    In addition $\dist{g(b, b)} \pbbisim \dist{g(c,c)} \pbbisim
    {\dist{\nullproc}} \ \npbbisim \ \dist{g}([0.5]
      \mkern1mu \dist{b} \oplus [0.5] \mkern1mu \dist{c}, [0.5] \mkern1mu \dist{b} \oplus
      [0.5] \mkern1mu \dist{c})$.  }
\end{example}

%% file: conclusion.tex
\section{Concluding remarks}
\label{sec:the_end}

In this paper we have presented the RBB safe specification format, to
the best of our knowledge, the first transition system specification
format in the quantitative setting that respects a weak equivalence.

Two main ideas underlie our approach.
First, the qualitative RBB safe specification format~\cite{Fok00}
occurs to translate smoothly to the way of specifying probabilistic
transitions systems as proposed in~\cite{DL12,LGD12}.
The representation of distributions over states via the algebraic of
distribution terms, is crucial to adapt the result.
With syntactic grip on distributions in place we are able to
incorporate the definitions of a nesting graph and a wild argument,
patience rules, 
w-nested context and w-nested position,
and finally the format specification itself.
%
In addition,
Lemma~\ref{lemma:congruence_closure_on_states_coincides_with_closure_on_distributions}
ensures that relations at the syntactic level lift well to the
semantical level.

Second, the characterization of branching bisimulation without
schedulers of \cite{AGT2012} reduced the complexity of our proofs.
In general, frameworks combining internal transitions and
probabilities, exploit schedulers to define weak combined transitions.
%
This gives extra overhead in the technical treatment.
Witnessing that working with these notions is not that easy, decision
algorithms for e.g.\ weak probabilistic bisimulation~\cite{HT12} and
weak distribution bisimulation~\cite{EHKTZ2013} are from recent years.

\newtext{
Future work includes an extension of the format, or a new one, 
to ensure that probabilistic branching bisimulation is a congruence. }
It also includes an extension to support quantitative premises.
As explained in Section~\ref{sec:format}, quantitative
premises cannot be used to include look-ahead in the format directly.
However, the two-sorted approach allows rules of the shape
\setlength{\abovedisplayskip}{10pt}
\setlength{\belowdisplayskip}{10pt}
\begin{displaymath}
  \dedrule{ \mu_1 \oplus_{p_1} \delta(\{y_2\}) (\{y_1\}) \gtgeq p_1
  \quad 
  \mu_2 \oplus_{p_2} \delta(\{y_1\}) (\{y_2\}) \gtgeq p_2
}{f(\mu_1, \mu_2) \trans[a] \theta}
\end{displaymath}
In order to achieve this goal, 
techniques introduced for the \ntmufxt\ format~\cite{DL12} can be used. 
%
%
Other research investigates the characterization of branching
bisimulation without schedulers to obtain new results in the
probabilistic context. Currently we are working on a logic to
characterize this relation, focusing on completeness.

\medskip

\emph{Acknowledgment.}
The authors would like to thank Pedro R. D'Argenio for helpful discussions.